\documentclass[
    reprint,
    aps,
    prl,
    twocolumn,
    nofootinbib
]{revtex4-2}
\usepackage{tikzit}

\tikzstyle{gate}=[shape=rectangle, text height=1.5ex, text depth=0.25ex, yshift=0.5mm, fill=white, draw=black, minimum height=3mm, yshift=-0.5mm, minimum width=3mm, font={\small}, tikzit category=circuit, inner sep=2pt]
\tikzstyle{big gate}=[shape=rectangle, text height=1.5ex, text depth=0.25ex, yshift=0.5mm, fill=white, draw=black, minimum height=10mm, yshift=-0.5mm, minimum width=5mm, font={\small}, tikzit category=circuit]
\tikzstyle{Z dot}=[inner sep=0mm, minimum size=2mm, shape=circle, draw=black, fill={rgb,255: red,221; green,255; blue,221}, tikzit category=zx]
\tikzstyle{Z phase dot}=[minimum size=5mm, font={\footnotesize\boldmath}, shape=rectangle, rounded corners=2mm, inner sep=0.2mm, outer sep=-2mm, scale=0.8, tikzit shape=circle, draw=black, fill={rgb,255: red,221; green,255; blue,221}, tikzit draw=blue, tikzit category=zx]
\tikzstyle{X dot}=[Z dot, shape=circle, draw=black, fill={rgb,255: red,255; green,136; blue,136}, tikzit category=zx]
\tikzstyle{X phase dot}=[Z phase dot, tikzit shape=circle, tikzit draw=blue, fill={rgb,255: red,255; green,136; blue,136}, font={\footnotesize\boldmath}, tikzit category=zx]
\tikzstyle{hadamard}=[fill=yellow, draw=black, shape=rectangle, inner sep=0.6mm, minimum height=1.5mm, minimum width=1.5mm, tikzit category=zx]
\tikzstyle{paulibox}=[fill={rgb,255: red,221; green,221; blue,255}, draw=black, shape=rectangle, inner sep=0.6mm, minimum height=5mm, minimum width=5mm, font={\footnotesize}, text height=1.5ex, text depth=0.25ex, tikzit category=zx]
\tikzstyle{vertex}=[inner sep=0mm, minimum size=1mm, shape=circle, draw=black, fill=black, tikzit category=misc]
\tikzstyle{vertex set}=[inner sep=0mm, minimum size=1mm, shape=circle, draw=black, fill=white, font={\footnotesize\boldmath}, tikzit category=misc]
\tikzstyle{small black dot}=[fill=black, draw=black, shape=circle, inner sep=0pt, minimum width=1.2mm, tikzit category=circuit]
\tikzstyle{cnot ctrl}=[fill=black, draw=black, shape=circle, inner sep=0pt, minimum width=1.2mm, tikzit category=circuit]
\tikzstyle{cnot targ}=[fill=white, draw=white, shape=circle, tikzit category=circuit, label={center:$\oplus$}, inner sep=0pt, minimum width=2.1mm, tikzit fill={rgb,255: red,102; green,204; blue,255}, tikzit draw=black]
\tikzstyle{ket}=[fill=white, draw=black, shape=regular polygon, regular polygon sides=3, regular polygon rotate=-30, scale=0.7, inner sep=1pt, tikzit category=circuit, tikzit shape=rectangle, tikzit fill=green]
\tikzstyle{bra}=[fill=white, draw=black, shape=regular polygon, regular polygon sides=3, regular polygon rotate=30, scale=0.7, inner sep=1pt, tikzit category=circuit, tikzit shape=rectangle, tikzit fill=red]
\tikzstyle{scalar}=[shape=rectangle, text height=1.5ex, text depth=0.25ex, yshift=0.5mm, fill=white, draw=black, minimum height=5mm, yshift=-0.5mm, minimum width=5mm, font={\small}]
\tikzstyle{clabel}=[fill=white, draw=none, shape=rectangle, tikzit fill={rgb,255: red,56; green,255; blue,242}, font={\footnotesize}, inner sep=1pt, tikzit category=labels]
\tikzstyle{empty diagram}=[draw={gray!40!white}, dashed, shape=rectangle, minimum width=1cm, minimum height=1cm, tikzit category=misc]
\tikzstyle{amap}=[fill=white, draw=black, shape=NEbox, tikzit category=asymmetric, tikzit fill=yellow, tikzit shape=rectangle]
\tikzstyle{amap conj}=[fill=white, draw=black, shape=NWbox, tikzit category=asymmetric, tikzit fill=green, tikzit shape=rectangle]
\tikzstyle{amap adj}=[fill=white, draw=black, shape=SEbox, tikzit category=asymmetric, tikzit fill=red, tikzit shape=rectangle]
\tikzstyle{amap trans}=[fill=white, draw=black, shape=SWbox, tikzit category=asymmetric, tikzit fill=orange, tikzit shape=rectangle]
\tikzstyle{astate}=[fill=white, draw=black, shape=NEtriangle, tikzit category=asymmetric, tikzit shape=circle, tikzit fill=yellow]
\tikzstyle{astate conj}=[fill=white, draw=black, shape=NWtriangle, tikzit category=asymmetric, tikzit shape=circle, tikzit fill=green]
\tikzstyle{astate adj}=[fill=white, draw=black, shape=SEtriangle, tikzit category=asymmetric, tikzit shape=circle, tikzit fill=red]
\tikzstyle{astate trans}=[fill=white, draw=black, shape=SWtriangle, tikzit category=asymmetric, tikzit shape=circle, tikzit fill=orange]
\tikzstyle{white dot}=[inner sep=0mm, minimum size=2mm, shape=circle, draw=black, fill={rgb,255: red,250; green,250; blue,250}]
\tikzstyle{white phase dot}=[minimum size=5mm, font={\footnotesize\boldmath}, shape=rectangle, rounded corners=2mm, inner sep=0.2mm, outer sep=-2mm, scale=0.8, tikzit shape=circle, draw=black, fill={rgb,255: red,250; green,250; blue,250}, tikzit draw=blue]
\tikzstyle{hbox}=[shape=rectangle, text height=2mm, fill={rgb,255: red,255; green,235; blue,61}, draw=black, minimum height=2mm, minimum width=2mm, font={\small}, tikzit category=zh, inner sep=0pt, rounded corners=0.5mm]
\tikzstyle{Z dot (zh)}=[inner sep=0mm, minimum size=2mm, shape=circle, draw=black, fill={rgb,255: red,250; green,250; blue,250}, tikzit category=zh]
\tikzstyle{X dot (zh)}=[Z dot, shape=circle, draw=black, fill={rgb,255: red,193; green,193; blue,193}, tikzit category=zh]
\tikzstyle{triangle}=[fill={rgb,255: red,255; green,136; blue,136}, draw=black, shape=isosceles triangle, isosceles triangle apex angle=60, minimum size=2.5mm, inner sep=0mm]
\tikzstyle{labelled hbox}=[shape=rectangle, text height=1.75ex, text depth=0.5ex, fill={rgb,255: red,255; green,235; blue,61}, draw=black, minimum height=3mm, minimum width=4mm, font={\small}, tikzit category=zh, inner sep=1.3pt, rounded corners=0.5mm]
\tikzstyle{Z phase dot (zh)}=[Z phase dot, tikzit shape=circle, tikzit draw=blue, fill={rgb,255: red,250; green,250; blue,250}, font={\footnotesize\boldmath}, tikzit category=zh]
\tikzstyle{X phase dot (zh)}=[Z phase dot, tikzit shape=circle, tikzit draw=blue, fill={rgb,255: red,193; green,193; blue,193}, font={\footnotesize\boldmath}, tikzit category=zh]
\tikzstyle{W node}=[fill=black, draw=black, shape=regular polygon, regular polygon sides=3, minimum size=2mm]
\tikzstyle{Z dot (zw)}=[fill=white, draw=black, shape=circle, minimum width=1.2mm, inner sep=0pt]
\tikzstyle{Z phase dot XL}=[Z phase dot, fill={rgb,255: red,250; green,250; blue,250}, draw=black, shape=circle, tikzit draw={rgb,255: red,191; green,0; blue,64}, tikzit shape=circle, font={\large\boldmath}, inner sep=0.0mm]

\tikzstyle{hadamard edge}=[-, dashed, dash pattern=on 2pt off 0.5pt, thick, draw={rgb,255: red,68; green,136; blue,255}]
\tikzstyle{box edge}=[-, dashed, dash pattern=on 2pt off 0.5pt, thick, draw={rgb,255: red,203; green,192; blue,225}]
\tikzstyle{brace edge}=[-, tikzit draw=blue, decorate, decoration={brace,amplitude=1mm,raise=-1mm}]
\tikzstyle{diredge}=[->, thick]
\tikzstyle{double edge}=[-, double, shorten <=-1mm, shorten >=-1mm, double distance=2pt]
\tikzstyle{gray edge}=[-, {gray!60!white}]
\tikzstyle{pointer edge}=[->, very thick, gray]
\tikzstyle{boldedge}=[-, line width=1.0pt, shorten <=-0.17mm, shorten >=-0.17mm]
\tikzstyle{bidir edge}=[<->, very thick, draw={rgb,255: red,191; green,191; blue,191}]
\tikzstyle{purple edge}=[->, thick, draw={rgb,255: red,225; green,117; blue,216}]
\tikzstyle{green edge}=[->, thick, draw={rgb,255: red,167; green,231; blue,137}]
\tikzstyle{orange edge}=[->, thick, draw={rgb,255: red,245; green,170; blue,63}]
\tikzstyle{any edge}=[->, thick, draw=cyan]
\tikzstyle{red edge}=[->, thick, draw={rgb,255: red,255; green,136; blue,136}]
\tikzstyle{bidiredge}=[<->, thick]
\tikzstyle{dashed diredge}=[->, dashed, dash pattern=on 1pt off 0.5pt]
\tikzstyle{bidashed diredge}=[<->, dashed, dash pattern=on 1pt off 0.5pt]
\tikzstyle{blue edge}=[-, draw={rgb,255: red,28; green,115; blue,237}]

\usepackage{graphicx}
\usepackage{bm}
\usepackage{physics}
\usepackage{amsthm}
\usepackage{cancel}
\usepackage{amssymb}
\usepackage{xspace}

\theoremstyle{definition}

\newtheorem{definition}{Definition}
\newtheorem{theorem}{Theorem}

\newtheorem{lemma}{Lemma}
\newtheorem{remark}{Remark}

\newcommand{\existinf}{\mathop{\exists}\limits^{\infty}}

\newcommand{\poly}{\mathrm{poly}}
\usepackage{booktabs}
\usepackage{svg}
\usepackage{multirow}
\usepackage[ruled,vlined,linesnumbered]{algorithm2e}

\usepackage{hyperref}
\usepackage[capitalize]{cleveref}

\hypersetup{colorlinks=true, citecolor=blue}

\newcommand{\pair}[1]{{\langle #1\rangle}}
\newcommand{\sat}{\texttt{SAT}\xspace}
\newcommand{\factoring}{\texttt{FACTORING}\xspace}
\newcommand{\fac}{\texttt{FAC}\xspace}
\newcommand{\many}{\leq_m^p}
\begin{document}

\title{Formal Framework for Quantum Advantage}
\author{Harry Buhrman, Niklas Galke, Konstantinos Meichanetzidis}

\affiliation{Quantinuum, Partnership House, Carlisle Place, London SW1P 1BX, United Kingdom}

\date{\today}

\begin{abstract}
Motivated by notions of quantum heuristics and by average-case rather than worst-case algorithmic analysis, we define quantum computational advantage in terms of {\em individual} problem instances. 
Inspired by the classical notions of Kolmogorov complexity and instance complexity, we define their quantum versions. 
This allows us to define {\em queasy} instances of computational problems, like e.g. Satisfiability and Factoring, as those whose quantum instance complexity is significantly smaller than their classical instance complexity. These instances indicate quantum advantage: they are easy to solve on a quantum computer, but classical algorithms struggle (they feel queasy).
Via a reduction from Factoring, we prove the existence of queasy Satisfiability instances; specifically, these instances are maximally queasy (under reasonable complexity-theoretic assumptions).
Further, we show that there is exponential algorithmic utility in the queasiness of a quantum algorithm.
This formal framework serves as a beacon that guides the hunt for quantum advantage in practice, and moreover, because its focus lies on single instances, it can lead to new ways of designing quantum algorithms.
\end{abstract}

\maketitle

While algorithms like Shor's for integer factoring \cite{Shor_1997} or the AJL algorithm for evaluating knot invariants \cite{aharonov2006polynomialquantumalgorithmapproximating} provide a clear roadmap to quantum advantage for specific problems, given reasonable theoretical assumptions and considering the classical state of the art, for many other hard problems, such a roadmap is not crisply known.
Even in the domain of quantum simulation, where quantum advantage is expected for strongly correlated many-body systems, since in general classical algorithms do not suffice, in practice, classical approximate methods can exhibit remarkable performance.
In short, the landscape of quantum advantage is vast and diverse \cite{huang2025vastworldquantumadvantage}.

This motivates a shift in focus from worst-case analysis of entire problem classes to the practical, instance-by-instance hunt for quantum advantage, and so, the notion of quantum heuristics (quristics) becomes important. Experimental, trial-and-error approaches to discovering effective quantum algorithms, much like how heuristics are used in classical computing. This work aims to establish the theoretical foundation for the age of quristics, which becomes increasingly relevant as quantum computers scale and practical applications are considered.

To do this, we introduce a quantum version of instance complexity, a concept built upon the foundations of Kolmogorov complexity. Algorithmic Information Theory, pioneered by Kolmogorov, Chaitin, and Solomonoff, defines the complexity of an individual object (like a string) as the length of the shortest program that can produce it \cite{Li2008}. Instance complexity, later developed by Orponen et al. \cite{orponen1994instance}, extends this idea to problem \emph{instances}, measuring the size of the smallest program that can correctly solve a specific instance of a larger problem.

We adapt these ideas to the quantum realm. Crucially, we use time-bounded complexity measures, which, unlike their plain counterparts, are \emph{computable} and directly relevant to practical computation where resources are finite. This choice allows us to formally define what makes a problem instance quantum-easy, or \emph{queasy}: one whose quantum instance complexity is significantly lower than its classical counterpart. By proving the existence of such queasy instances within the canonical NP-complete problem of Boolean satisfiability (\sat), we demonstrate that the search for quantum advantage can be rigorously defined and guided. Though our definitions are abstract, they serve as a north star for the practical search for quantum advantage, providing a formal language to characterise and identify specific instances as promising.


{\bf Queasy Instances --}
A paradigmatic example of a problem class that contains queasy instances is \texttt{FACTORING}.
Shor's algorithm provides a $O(n^3)$ solution to obtaining the prime factors of a given $n$-bit integer, whereas the best classical algorithm, general number field sieve, runs in subexponential time $\Theta( \exp( n^{1/3} (\log n)^{2/3} ) )$ \cite{GNFS_Book_1993}.
Note that this is a problem that is also efficiently verifiable classically.
This is not believed to be the case for most problems in BQP.

There exist problems for which it is not expected that quantum computers will provide an advantage, such as \texttt{SAT} or Local Hamiltonian (Ground state finding, which is QMA-hard, the quantum analogue of NP), as all algorithms, classical and quantum, are expected to be `bad'.
And for classically `easy' problem classes, most famously P, again, we do not expect quantum computers to provide a speedup, as classical algorithms are already `good'.
But as long as we believe P$\subsetneq$BQP and NP$\neq$BQP, we expect queasy instances to exist inside BQP.

Note: throughout this work, we use $\eqsim$, $\gtrsim$, and $\lesssim$ for equalities and inequalities that hold up to a constant.


Kolmogorov complexity was introduced as a measure of the randomness of strings, also known as Algorithmic Information \cite{Li2008}. It presents an alternative point of view to the information-theoretic Shannon Entropy \cite{shannon}, as it does not assume a probability distribution over strings, but regards an \emph{individual} given string.

The `plain' Kolmogorov complexity $C(x) = |P|$ of an $n$-bit string $x$ is the size of the shortest program $P$ that runs on a universal Turing machine and prints the string $x$.
It is independent (up to a constant) of the choice of universal Turing machine,
$1 \leq C(x) \lesssim n $. Specifically, for simple strings like all zeroes we have
$C(0^n) \leq \log n$, there are also complex, or \emph{incompressible}, strings. Finally, $C$ is uncomputable, as the search for the smallest program reduces to the Halting problem.

Since we are interested in the runtime of algorithms, it is useful to define the \emph{time-bounded} version of Kolmogorov complexity,
which, importantly, is \emph{computable}.

\begin{definition} \label{def:KC}
The time-bounded Kolmogorov complexity $C^t(x) = |P|$ of an $n$-bit string $x$ is the size of the shortest program $P$ that runs in time $t(n)$ and returns $x$.
\end{definition}

We now define the \emph{quantum version} of $C^t$, taking into account that quantum computation is a \emph{probabilistic} model of computation.
In this context, we consider a classical program $P$ that outputs a quantum circuit $U$, ie the \emph{classical description} of a quantum program which produces the given string.

\begin{definition} \label{def:QKC}
    The time-bounded Quantum Kolmogorov complexity $QC^{t,\varepsilon}(x) = |P_U|$ of an $n$-bit string $x$ is the size of the shortest program $P_U$ that runs in time $t(n)$ and generates a quantum circuit $U$ that outputs $x$ with probability $\varepsilon>0$.
\end{definition}

The quantum and classical complexities are related as follows.
\begin{remark}
$C(x) \lesssim QC^{t',\varepsilon}(x) \lesssim C^t(x)$, where $t'=t\log t$.
\end{remark}
This means that time-bounding the Kolmogorov complexity can only increase it, as we are restricting the pool of programs allowed.
Furthermore, the quantum time-bounded complexity is smaller than the classical, because any classical computation is also a quantum computation.
However, note that this incurs a time overhead from converting a Turning machine to a reversible circuit \cite{Nielsen_Chuang_2010}.





\cref{def:QKC} can be made independent of $\varepsilon$:


\begin{theorem} \label{thm:noepsilon}
    For $\varepsilon>0$,
    $QC^{nt, 1-r^{-n}}(x)\lesssim QC^{t,\varepsilon}(x)$
    for some $r=r(\varepsilon) > 0$ and large enough $n\in\mathbb{N}$. (Proof in Appendix)
\end{theorem}

So, we can remove the dependence on the probability $\epsilon$ and simply write $QC^t$ for the time-bounded quantum Kolmogorov complexity.


Furthermore, we shall make use of a weaker quantity than the Kolmogorov complexity, namely the \emph{Distinguishing Complexity}, $CD(x)$. This is the complexity of identifying, or \emph{recognising}, a string $x$.
It is the length of the shortest program that accepts \emph{only} string $x$, but does not necessarily generate $x$.
And, more specifically, we use its time-bounded version.

\begin{definition}
    The time-bounded Distinguishing Complexity $CD^t(x)$ is the size of the shortest program $P : P(x)=1, P(y\neq x) = 0$ that runs in time $t$.
\end{definition}

The quantum version of time-bounded Distinguishing complexity is naturally defined as follows.

\begin{definition}
    The time-bounded Quantum Distinguishing Complexity $QCD^t(x)=|P_U|$ is the size of the shortest program $P_U$ that runs in time $t$ and generates a quantum circuit $U$ that accepts with probability $> 1/2+\epsilon$, for some $\epsilon>0$, and for $y\neq x$ it accepts with probability $<1/2-\epsilon$.
\end{definition}


Having defined the complexity of a string, both classical and quantum, we turn to the definition of the complexity of a particular \emph{instance} of a problem. This was formally done in~\cite{orponen1994instance}.

Consider a computational \emph{problem}, or \emph{language}, $\texttt{L}\subseteq\{ 0,1 \}^*$,
where problem instances $x$ are represented as $n$-bit strings.
As we are interested in time-bounded notions of complexity, we allow algorithms, or programs, to return "I don't know", denoted $\bot$, as a third option to $\{0,1\}$, given an instance as input.
In particular, we are interested in well-behaved programs with respect to the given problem.
This means that given an instance $x$, a program $P$ must be correct when it does not return "I don't know", in other words, it must be consistent with the language.

\begin{definition}
A program $P$ running in time $t(n)$ is \texttt{L}-consistent if $P(x)\in\{0,1,\bot\}$, $\forall x$, and if $P(x)\neq \bot$ then $P(x)=\chi_\texttt{L}(x)$.
\end{definition}

Here, $\chi_\texttt{L}(x)$ is the characteristic function, $\chi_\texttt{L}(x)=1$ if $x\in\texttt{L}$ and $\chi_\texttt{L}(x)=0$ if $x\notin\texttt{L}$.
A trivial example of a consistent problem (with any \texttt{L} problem) is the program that always returns $\bot$ independently of input instance.

We now define the time-bounded \emph{Instance complexity} with respect to a given computational problem.

\begin{definition}
    The time-bounded instance complexity $ic^t (x:L) = |P|$ of a problem instance $x$ is the size of the shortest \texttt{L}-consistent program $P : P(x)\neq \bot$ that runs in time $t(n)$.
\end{definition}

Technically, the time-bounded instance complexity is not decidable because of the consistency property. However, approximating it within an additive $\log(n)$ term makes it decidable. The time-bounded Kolmogorov complexity upper bounds the Instance complexity.

\begin{remark} \label{rem:icleqC}
The time-bounded instance complexity of $x$ with respect to any problem \texttt{L} is upper bounded by the time-bounded Kolmogorov complexity of the instance, 
$ic^{t'}(x:L) \leq C^t(x)$, where $t'=t\log t + n$. (Proof in Appendix)
\end{remark}





In general, identifying a string is simpler than generating it, so the Distinguishing Complexity is upper bounded by the Kolmogorov Complexity.
Therefore, it constitutes a stronger upper bound to the instance complexity.

\begin{remark}\label{rem:icleqCD}
    $ic^t(x:\texttt{L}) \lesssim CD^t(x) \lesssim C^t(x), \forall x,\texttt{L}$. (Proof in Appendix)
\end{remark}

Among instances in the problem class, we can identify the \emph{hard} and \emph{easy} instances, in analogy with \emph{incompressible} and \emph{compressible} strings.

\begin{definition} Instance $x$ is hard iff $ic^t(x:\texttt{L}) \eqsim CD^t(x)$.    
\end{definition}


\begin{definition} Instance $x$ is easy iff $ic^t(x:\texttt{L}) \ll CD^t(x)$.    
\end{definition}

And, naturally, we can define the corresponding quantum versions of consistency with a language and instance complexity.

\begin{definition}


A program $P_U$ is quantum $\epsilon$-\texttt{L}-consistent if it generates a circuit $U$ that uses a dedicated qubit $q_0$ to signal its confidence. If the outcome of $q_0$ indicates "I know", ie $p(q_0=1)>1/2+\epsilon$, then a second qubit $q_1$ must yield the correct answer with high probability, ie $p(q_1=\chi_\texttt{L}(x))>1/2+\epsilon$. If $q_0$ indicates "I don't know", ie $p(q_0=0)>1/2+\epsilon$, the outcome of $q_1$ is disregarded.
\end{definition}

\begin{definition}
    The time-bounded Quantum Instance Complexity $Qic^t (x:\texttt{L})$ of an $n$-bit instance $x$ is the size of the shortest quantum-$\epsilon$\texttt{-L}-consistent program $P_U$ that runs in time $t(n)$ and  decides $x$, for some\footnote{The choice of $\epsilon>0$ is independent of the complexity measure as simple repeating the program and taking the majority value will increase $\epsilon$.} $\epsilon > 0$.
\end{definition}

Similarly to the classical case, we have that quantum instance complexity is bounded by quantum Kolmogorov complexity.

\begin{remark} \label{rem:qicleqqcd}
$Qic^{t'}(x:\texttt{L})\leq QCD^t(x)$. (Proof similar to \cref{rem:icleqCD})
\end{remark}


If a problem \texttt{L} is in BQP, then for any instance $x$ its quantum instance complexity $Qic^\mathrm{poly}(x:L)$ is upper-bounded by a constant (the size of the program implementing the BQP algorithm). If \texttt{L} $\notin$ P, then its classical instance complexity, $ic^\poly(x:L)$, cannot be upper-bounded by a constant (otherwise we could find a polytime program implying \texttt{L} $\in$ P). Therefore, $BQP \neq P$ directly implies the existence of instances where $ic^{>\poly}(x:\texttt{L})$ is large and $Qic^\poly(x:\texttt{L})$ is small — ie the existence of \emph{queasiness}. Including a brief proof of this would be a strong start to the section.

The definition of classically hard and queasy instances regards the relation between quantum and classical instance complexities. These are instances for which a quantum computational advantage is expected.
The theoretical framework adopted in this work defines quantum advantage in an instance-based fashion, having in mind quristics.

Firstly, we have that the classical instance complexity bounds the quantum instance complexity, similarly to how classical Kolmogorov complexity bounds quantum Kolmogorov complexity.

\begin{remark}
$Qic^{t'}(x:\texttt{L}) \leq ic^t(x:\texttt{L})$, where $t' = t \log{t}$.    
\end{remark}

A queasy instance is now defined by the difference between time-bounded classical and quantum instance complexities.

\begin{definition}
    An instance $x$ is \emph{Queasy} with respect to problem \texttt{L} iff $Qic^{t'}(x:\texttt{L}) \lesssim ic^t(x:\texttt{L})$, with $t'=\poly(n)$ and $t<\exp(n)$, where $n=|x|$.
\end{definition}

Note, we have restricted the quantum time bound to polynomial, whereas we allow a sub exponential classical time bound.
We do not allow an exponential classical time bound, since then the classical algorithm may simply be a simulation of the quantum algorithm.
This definition also naturally incorporates the cases where quantum algorithms and in general quantum time-evolutions happen to be efficiently simulable, for example Clifford circuits, free-fermionic dynamics, or other approximate methods, for example, those based on tensor networks, that take advantage of the entanglement structure or the presence of noise.
Differentiating between quantum and classical time is important, as in practice it may be the case that an exponentially scaling classical algorithm may have a smaller time to solution than a polynomially scaling quantum algorithm, for a given instance of a given size.

We can now define a measure the \emph{queasiness} of an instance, ie how much easier it seems to a quantum algorithm versus a classical one, or how `queasy' it makes the classical algorithm `feel'.

\begin{definition}
The queasiness of an instance $x$ is the difference between the classical and quantum instance complexities $\Delta ic^{t,t'} (x:\texttt{L}) := ic^t(x:\texttt{L}) - Qic^{t'}(x:\texttt{L})$, with respect to the classical time-bound $t$ and the quantum time-bound $t'$.
\end{definition}

Note that the search for quantum advantage entails identifying instances exhibiting \emph{maximal queasiness}, ie when the bound $\Delta ic^{t,t'}(x:\texttt{L}) \lesssim CD^t(x)$ is saturated.

{\bf Algorithmic Utility --}
Now, since $ic^t(x:\texttt{L}) \lesssim CD^t(x)$, consider the case when $ic^t(x:\texttt{L}) \ll CD^t(x)$, where $P$ is the shortest L-consistent program that decides whether $x\in \texttt{L}$ in time $t$,
and consider the set $S$ of problem instances that are also solved with the same program $P$.
We can show the following:

\begin{theorem}\label{thm:usefulness}
   Assume $t' = t +p$ for some polynomial $p$. For any instance $x$ of length $n$ and problem \texttt{L}, assume $CD^{t'}(x) - ic^t(x:\texttt{L}) = d$, where $P(x)=\chi_\texttt{L}(x)$ and  $|P|=ic^t(x:\texttt{L})$. Consider the set $S=\{y|P(y)\neq \bot \}$. Then, $|S|\geq 2^{(d/2) -c} - O(n^{c'})$, for some constants $c$ and $c'$. (Proof in Appendix)
\end{theorem}

The cardinality of $S$ scales with $2^{d/2}$, implying that $P$ is useful for \emph{exponentially} many other instances.
\cref{thm:usefulness} is also true in the quantum case with exactly the same proof for the time-bounded Quantum Distinguishing complexity.
This means that the smaller the instance complexity, assuming the distinguishing complexity is high, the higher the \emph{algorithmic utility}, in the sense that the quantum program also solves \emph{exponentially more} other instances rather than returning $\bot$.
Especially, if $\Delta ic^{t,t'}(x:L) > 0$, then this would mean that a quantum algorithm that decides instance $x$ with respect to language $\texttt{L}$ would be more useful than any classical algorithm with the same runtime $t$, not only in terms of quantum advantage for the instance $x$ but also in the sense that it could decide \emph{exponentially} in $\Delta ic^{t,t'}(x:L)$ more instances than any classical algorithm could (without necessarily being the shortest possible program for them).
There is one notable exception, namely when the difference in quantum and classical instance complexities ($Qic^{t'}(x:\texttt{L}) < ic^t(x:\texttt{L})$) is due to the difference in quantum and classical distinguishing complexities ($QCD^{t'}(x)< CD^{t}(x)$).
In other words, the above intuition holds when the quantum and classical distinguishing complexities for this instance are close ($QCD^{t'}(x) \simeq CD^{t}(x)$). And note that the distinguishing complexities are independent of the problem (or language) at hand.


{\bf Queasy SAT --} We now turn our attention to \sat the well known NP-complete problem and examine whether it contains queasy instances. We show that this is indeed the case: there exist \emph{infinitely many} formulae $\phi\in \sat$, that are close to being \emph{maximally queasy}.
This shows that the definitions we introduced are useful for identifying instances with quantum advantage.
In fact, we identify a meaningful division for \sat into three types of instances: the classically easy ones, the ones hard for both classical and quantum, and the queasy ones. Let's first turn our attention to the queasy instances for \sat:


\begin{theorem}\label{thm:infqueasy}
    If \factoring on $n$-bit instances requires $\Omega(2^{n^\epsilon})$ time\footnote{The best know algorithm requires $\epsilon \geq 1/3$}, for some $\epsilon > 0$, then $\existinf \phi \in \texttt{SAT} : \Delta ic^{t, n^3}(\phi:\texttt{SAT}) \geq n^\delta$, for some $\delta > 0$ and $t(n) \in O(2^{n^\delta})$.
\end{theorem}

To prove \cref{thm:infqueasy}, we first show that its statement is true for a suitable version of \texttt{FACTORING}, which we call \fac.

\begin{definition}
$\texttt{FAC} = \{ \pair{x,a} :$ the largest prime factor of $x$ has $a$ as~prefix$\footnote{This means that if the $p$ is the largest prime factor of $x$, then $\pair{x,a}$ is in \texttt{FAC} if the binary representation of $p = ab$, that is, $p$ starts with the binary string $a$ concatenated with the binary string $b$ ($a$ is a prefix of $p$).}  \}$  $\in$ NP $\cap$ co-NP. 
\end{definition}

 The following lemma shows that the \texttt{FAC} problem contains \emph{infinitely many} queasy instances. 
 
\begin{lemma}\label{lemma:pruning}
 If \factoring on $n$-bit instances requires $\Omega(2^{n^\epsilon})$ time, for some $\epsilon > 0$, then $\existinf z \in \fac : \Delta ic^{t,n^3}(z:\fac) \geq n^\delta$, for  $\delta < \epsilon$ and $t(n) \in O(2^{n^\delta})$. (Proof in Appendix)
\end{lemma}

Note that $\texttt{FAC} \in$ BQP and hence: $\forall z:  $ $Qic^{t}(z:\texttt{FAC})\leq O(1)$ for $t = n^3$. This means that there exists a constant-size description of a quantum algorithm that efficiently solves the problem, namely Shor's algorithm \cite{Shor_1997}. So it suffices to show that there exist infinitely many instances in \texttt{FAC} that have high classical instance complexity.

To proceed, we need the following:

\begin{lemma}\label{lem:easytoeasyhardtohard}
If A $\many$ B via a polynomial time computable function $f$ then for some polynomials $t$ and $t'$:
\begin{enumerate}
    \item $ic^{t}(x:\texttt{A}) \lesssim ic^{t'}(f(x):\texttt{B})$.
    \item $Qic^{t}(x:\texttt{A}) \gtrsim Qic^{t'}(f(x):\texttt{B})$ if in addition  $f$ is poly-time invertible and 1-1.
\end{enumerate}
\end{lemma}

This implies that under a poly-time reduction from language \texttt{A} to language \texttt{B}, hard classical instances map to hard and quantum easy instances remain easy, and thus that queasy instances for \texttt{A} get translated to queasy instances for \texttt{B}.
Note that invertibility is needed to map the structure of \fac into the \sat formula. This structure is necessary and often ignored when one naively searches for quantum algorithms under the variational paradigm.

It is not hard to construct a poly-time and invertible  1-1 reduction from \fac to \sat by augmenting the standard Cook-Levin construction with a description of the instance reduced from. Putting everything together leads to the proof of \cref{thm:infqueasy}.



Having proved the existence of infinitely many queasy \sat instances, we also investigate the quantum hard \sat instances (and thus also classically hard).
A first observation is that if \sat $\notin$ BQP then \sat contains infinitely many instances that have more than constant quantum instance complexity: $Qic^{t}(x:\texttt{L}) > c$, for every constant $c$ depending on \texttt{L}, and $t$ any polynomial.
This is because a set \texttt{L} $\in$ BQP iff $\forall x: Qic^{t}(x:\texttt{L}) \leq c$ for some constant $c$ and polynomial $t$. For \sat, this can be improved to logarithmic if NP $\not \subseteq$ BQP and even to linear under Quantum Strong Exponential Time Hypothesis (QSETH)~\cite{buhrman_et_al:LIPIcs.STACS.2021.19,aaronson_et_al:LIPIcs.CCC.2020.16}, which essentially says there is no quantum algorithm for \sat that runs in time $2^{(\frac{1}{2}  -\epsilon)n}$, for $\epsilon>0$.

\begin{theorem} The following hold for infinitely many $x$:
\begin{enumerate}
\item If NP $\not \subseteq$ BQP then $Qic^{t}(x:\texttt{SAT}) \geq  \omega(\log{|x|})$, for $t$ any polynomial.
\item If QSETH is true then  $Qic^{t}(x:\texttt{SAT}) \geq 
    |x|\delta$, for $t(n)=2^{n\gamma}$, with $\delta >0 , \gamma > 0$, and $2\delta + \gamma < \frac{1}{2}$.
\end{enumerate}
\end{theorem}

The proofs follow a similar pruning algorithm as the proof of \cref{lemma:pruning}. See also \cite{orponen1994instance} for similar results in the classical setting.

The final result in this section shows that under the assumption that co-NP $\nsubseteq$ QCMA/poly (which is a stronger assumption than NP $\not \subseteq$ BQP) there exists an exponentially dense set of hard instances for \texttt{SAT}.

\begin{theorem}
If co-NP $\nsubseteq$ QCMA/poly then there exists a $\delta > 0$ such that for infinitely many $n$ the set $H^{\leq n} = \{ x \mathrel{| Qic^{t}(x:\texttt{SAT}) \geq 
    |x|^\delta} \}$ has density $2^{n^\delta}$, for $\delta>0$ and $t$ any polynomial. (Proof as in Ref~\cite{BH08})
\end{theorem}

The class QCMA/poly contains decision problems for which a classical proof can be efficiently verified by a quantum computer, provided the computer also receives a polynomial-sized advice string that depends only on the length of the input.


{\bf Advantage Landscape --}
Through this work, we have provided a theoretical foundation for classifying computational problems into easy (for both quantum and classical computers), queasy (quantum easy and classically hard), and hard (for both quantum and classical computers), enabling the mapping of the instance landscape.
Consider the \emph{queasiness factor}, defined in terms of time-bounded quantum and classical instance complexities for a given problem instance $x$ with respect to a language \texttt{L} as:

\begin{definition} \label{def:queasinessfactor}
$Ric^{t,t'}(x:\texttt{L}) = 1 - \frac{Qic^{t'}(x:L)}{ic^{t}(x:L)} \in [0,1)$.
\end{definition}

The queasiness factor vanishes for both easy and hard instances but approaches $1$ for maximally queasy instances.
We have proved above that under reasonable assumptions, hard \texttt{SAT} instances are exponentially dense.
Further, when $ic<<C$, easy instances are also exponentially dense.
An important open question that we pose regards the \emph{queasy instance density}.

Note that in contrast to Levin's time-bounded complexity, which simultaneously minimises both description length and runtime ($|P|+\log(t_P)$) \cite{LEVIN198415}, our \emph{external} time bounds provide the flexibility to compare the functionally different classical and quantum computations.

{\bf Discussion --}
Motivated by the practical, heuristic-driven search for quantum advantage, we have established a formal framework for instance-based quantum advantage using the language of instance complexity, defining "queasy" instances as those for which a quantum computer requires a fundamentally simpler program than a classical computer. We proved the existence of infinitely many queasy instances within \sat by reduction from \factoring, demonstrating in-principle quantum advantage for instances of a classic NP-complete problem. This instance-based advantage implies a greater algorithmic utility; a queasy instance points to a quantum heuristic that solves an exponentially larger set of problems than its classical counterpart.

Moving from theory to practice, we acknowledge that the asymptotic time bounds $t(n)$ used in complexity theory omit crucial details like constant factors and the enormous disparity in clock speeds between current quantum and classical hardware, which determine the concrete time-to-solution. Furthermore, our definitions, which rely on achieving a success probability greater than $1/2+\epsilon$, formally apply to the fault-tolerant era, where quantum error correction can suppress physical gate noise to manageable levels.
For NISQ devices without effective QEC, significant gate noise may prevent an algorithm from reliably satisfying this condition. If the success probability cannot be amplified above the $1/2$ threshold, then the quantum complexities as we have defined them are not applicable.
Furthermore, many practical problems also require computing a numerical quantity to a certain precision, or sample from a distribution, rather than solving an exact decision problem; our string-based framework could be extended to these cases which we leave for future work.
Despite these theoretical abstractions, our framework provides a practical methodology for an empirical search for queasy instances. One may use practical proxies of quantities, like approximating distinguishing complexity using lossless compressors, or instance complexity via rigorous resource estimation pipelines (as performed in Ref.~\cite{laakkonen2025quantumadvantageendtoendquantum}). This enables an empirical program of sampling instances from a problem class and mapping out an advantage landscape using the normalised queasiness factor.

This hardware-centric, experimental approach is analogous to the recent history of artificial intelligence, where the advent of GPUs enabled a new era of trial-and-error discovery with long-known concepts like neural networks. As quantum hardware matures, we envision a similar age of quristics. 
In general, one would fix an instance, and then fix a constant time bound $c$ (time-budget), and empirically estimate $Ric^{c,c}$ (using state-of-the-art algorithms) to characterise the queasiness of that instance.
Ultimately, we envision that these foundational tools will not only guide the hunt for quantum advantage, both in theory and in practice, but can also be used to develop novel quantum algorithms for tasks like data compression and protocols for verification of quantum computers.

\bibliographystyle{eptcs}
\bibliography{refs}

\begin{thebibliography}{10}
\providecommand{\bibitemdeclare}[2]{}
\providecommand{\surnamestart}{}
\providecommand{\surnameend}{}
\providecommand{\urlprefix}{Available at }
\providecommand{\url}[1]{\texttt{#1}}
\providecommand{\href}[2]{\texttt{#2}}
\providecommand{\urlalt}[2]{\href{#1}{#2}}
\providecommand{\doi}[1]{doi:\urlalt{https://doi.org/#1}{#1}}
\providecommand{\eprint}[1]{arXiv:\urlalt{https://arxiv.org/abs/#1}{#1}}
\providecommand{\bibinfo}[2]{#2}

\bibitemdeclare{article}{Shor_1997}
\bibitem{Shor_1997}
\bibinfo{author}{Peter~W. \surnamestart Shor\surnameend}
  (\bibinfo{year}{1997}): \emph{\bibinfo{title}{Polynomial-Time Algorithms for
  Prime Factorization and Discrete Logarithms on a Quantum Computer}}.
\newblock {\slshape \bibinfo{journal}{SIAM Journal on Computing}}
  \bibinfo{volume}{26}(\bibinfo{number}{5}), p. \bibinfo{pages}{1484–1509},
  \doi{10.1137/s0097539795293172}.
\newblock \urlprefix\url{http://dx.doi.org/10.1137/S0097539795293172}.

\bibitemdeclare{misc}{aharonov2006polynomialquantumalgorithmapproximating}
\bibitem{aharonov2006polynomialquantumalgorithmapproximating}
\bibinfo{author}{Dorit \surnamestart Aharonov\surnameend},
  \bibinfo{author}{Vaughan \surnamestart Jones\surnameend} \&
  \bibinfo{author}{Zeph \surnamestart Landau\surnameend}
  (\bibinfo{year}{2006}): \emph{\bibinfo{title}{A Polynomial Quantum Algorithm
  for Approximating the Jones Polynomial}}.
\newblock \eprint{quant-ph/0511096}.

\bibitemdeclare{misc}{huang2025vastworldquantumadvantage}
\bibitem{huang2025vastworldquantumadvantage}
\bibinfo{author}{Hsin-Yuan \surnamestart Huang\surnameend},
  \bibinfo{author}{Soonwon \surnamestart Choi\surnameend},
  \bibinfo{author}{Jarrod~R. \surnamestart McClean\surnameend} \&
  \bibinfo{author}{John \surnamestart Preskill\surnameend}
  (\bibinfo{year}{2025}): \emph{\bibinfo{title}{The vast world of quantum
  advantage}}.
\newblock \eprint{2508.05720}.

\bibitemdeclare{inbook}{Li2008}
\bibitem{Li2008}
\bibinfo{author}{Ming \surnamestart Li\surnameend} \& \bibinfo{author}{Paul
  \surnamestart Vitányi\surnameend} (\bibinfo{year}{2008}):
  \emph{\bibinfo{title}{Preliminaries}}, p. \bibinfo{pages}{1–99}.
\newblock \bibinfo{publisher}{Springer New York},
  \doi{10.1007/978-0-387-49820-1_1}.
\newblock \urlprefix\url{http://dx.doi.org/10.1007/978-0-387-49820-1_1}.

\bibitemdeclare{article}{orponen1994instance}
\bibitem{orponen1994instance}
\bibinfo{author}{Pekka \surnamestart Orponen\surnameend},
  \bibinfo{author}{Ker-I \surnamestart Ko\surnameend}, \bibinfo{author}{Uwe
  \surnamestart Schöning\surnameend} \& \bibinfo{author}{Osamu \surnamestart
  Watanabe\surnameend} (\bibinfo{year}{1994}): \emph{\bibinfo{title}{Instance
  Complexity}}.
\newblock {\slshape \bibinfo{journal}{Journal of the ACM}}
  \bibinfo{volume}{41}(\bibinfo{number}{1}), pp. \bibinfo{pages}{96--121},
  \doi{10.1145/174644.174648}.

\bibitemdeclare{book}{GNFS_Book_1993}
\bibitem{GNFS_Book_1993}
 (\bibinfo{year}{1993}): \emph{\bibinfo{title}{The development of the number
  field sieve}}.
\newblock \bibinfo{publisher}{Springer Berlin Heidelberg},
  \doi{10.1007/bfb0091534}.
\newblock \urlprefix\url{http://dx.doi.org/10.1007/BFb0091534}.

\bibitemdeclare{article}{shannon}
\bibitem{shannon}
\bibinfo{author}{Claude~Elwood \surnamestart Shannon\surnameend}
  (\bibinfo{year}{1948}): \emph{\bibinfo{title}{A Mathematical Theory of
  Communication}}.
\newblock {\slshape \bibinfo{journal}{The Bell System Technical Journal}}
  \bibinfo{volume}{27}, pp. \bibinfo{pages}{379--423}.
\newblock
  \urlprefix\url{http://plan9.bell-labs.com/cm/ms/what/shannonday/shannon1948.pdf}.

\bibitemdeclare{book}{Nielsen_Chuang_2010}
\bibitem{Nielsen_Chuang_2010}
\bibinfo{author}{Michael~A. \surnamestart Nielsen\surnameend} \&
  \bibinfo{author}{Isaac~L. \surnamestart Chuang\surnameend}
  (\bibinfo{year}{2010}): \emph{\bibinfo{title}{Quantum Computation and Quantum
  Information: 10th Anniversary Edition}}.
\newblock \bibinfo{publisher}{Cambridge University Press}.

\bibitemdeclare{inproceedings}{buhrman_et_al:LIPIcs.STACS.2021.19}
\bibitem{buhrman_et_al:LIPIcs.STACS.2021.19}
\bibinfo{author}{Harry \surnamestart Buhrman\surnameend},
  \bibinfo{author}{Subhasree \surnamestart Patro\surnameend} \&
  \bibinfo{author}{Florian \surnamestart Speelman\surnameend}
  (\bibinfo{year}{2021}): \emph{\bibinfo{title}{{A Framework of Quantum Strong
  Exponential-Time Hypotheses}}}.
\newblock In \bibinfo{editor}{Markus \surnamestart Bl\"{a}ser\surnameend} \&
  \bibinfo{editor}{Benjamin \surnamestart Monmege\surnameend}, editors:
  {\slshape \bibinfo{booktitle}{38th International Symposium on Theoretical
  Aspects of Computer Science (STACS 2021)}}, {\slshape
  \bibinfo{series}{Leibniz International Proceedings in Informatics (LIPIcs)}}
  \bibinfo{volume}{187}, \bibinfo{publisher}{Schloss Dagstuhl --
  Leibniz-Zentrum f{\"u}r Informatik}, \bibinfo{address}{Dagstuhl, Germany},
  pp. \bibinfo{pages}{19:1--19:19}, \doi{10.4230/LIPIcs.STACS.2021.19}.
\newblock
  \urlprefix\url{https://drops.dagstuhl.de/entities/document/10.4230/LIPIcs.STACS.2021.19}.

\bibitemdeclare{inproceedings}{aaronson_et_al:LIPIcs.CCC.2020.16}
\bibitem{aaronson_et_al:LIPIcs.CCC.2020.16}
\bibinfo{author}{Scott \surnamestart Aaronson\surnameend},
  \bibinfo{author}{Nai-Hui \surnamestart Chia\surnameend},
  \bibinfo{author}{Han-Hsuan \surnamestart Lin\surnameend},
  \bibinfo{author}{Chunhao \surnamestart Wang\surnameend} \&
  \bibinfo{author}{Ruizhe \surnamestart Zhang\surnameend}
  (\bibinfo{year}{2020}): \emph{\bibinfo{title}{{On the Quantum Complexity of
  Closest Pair and Related Problems}}}.
\newblock In \bibinfo{editor}{Shubhangi \surnamestart Saraf\surnameend},
  editor: {\slshape \bibinfo{booktitle}{35th Computational Complexity
  Conference (CCC 2020)}}, {\slshape \bibinfo{series}{Leibniz International
  Proceedings in Informatics (LIPIcs)}} \bibinfo{volume}{169},
  \bibinfo{publisher}{Schloss Dagstuhl -- Leibniz-Zentrum f{\"u}r Informatik},
  \bibinfo{address}{Dagstuhl, Germany}, pp. \bibinfo{pages}{16:1--16:43},
  \doi{10.4230/LIPIcs.CCC.2020.16}.
\newblock
  \urlprefix\url{https://drops.dagstuhl.de/entities/document/10.4230/LIPIcs.CCC.2020.16}.

\bibitemdeclare{inproceedings}{BH08}
\bibitem{BH08}
\bibinfo{author}{Harry \surnamestart Buhrman\surnameend} \&
  \bibinfo{author}{John~M. \surnamestart Hitchcock\surnameend}
  (\bibinfo{year}{2008}): \emph{\bibinfo{title}{NP-Hard Sets Are Exponentially
  Dense Unless coNP C NP/poly}}.
\newblock In: {\slshape \bibinfo{booktitle}{2008 23rd Annual IEEE Conference on
  Computational Complexity}}, pp. \bibinfo{pages}{1--7},
  \doi{10.1109/CCC.2008.21}.

\bibitemdeclare{article}{LEVIN198415}
\bibitem{LEVIN198415}
\bibinfo{author}{Leonid~A. \surnamestart Levin\surnameend}
  (\bibinfo{year}{1984}): \emph{\bibinfo{title}{Randomness conservation
  inequalities; information and independence in mathematical theories}}.
\newblock {\slshape \bibinfo{journal}{Information and Control}}
  \bibinfo{volume}{61}(\bibinfo{number}{1}), pp. \bibinfo{pages}{15--37},
  \doi{https://doi.org/10.1016/S0019-9958(84)80060-1}.
\newblock
  \urlprefix\url{https://www.sciencedirect.com/science/article/pii/S0019995884800601}.

\bibitemdeclare{misc}{laakkonen2025quantumadvantageendtoendquantum}
\bibitem{laakkonen2025quantumadvantageendtoendquantum}
\bibinfo{author}{Tuomas \surnamestart Laakkonen\surnameend},
  \bibinfo{author}{Enrico \surnamestart Rinaldi\surnameend},
  \bibinfo{author}{Chris~N. \surnamestart Self\surnameend},
  \bibinfo{author}{Eli \surnamestart Chertkov\surnameend},
  \bibinfo{author}{Matthew \surnamestart DeCross\surnameend},
  \bibinfo{author}{David \surnamestart Hayes\surnameend},
  \bibinfo{author}{Brian \surnamestart Neyenhuis\surnameend},
  \bibinfo{author}{Marcello \surnamestart Benedetti\surnameend} \&
  \bibinfo{author}{Konstantinos \surnamestart Meichanetzidis\surnameend}
  (\bibinfo{year}{2025}): \emph{\bibinfo{title}{Less Quantum, More Advantage:
  An End-to-End Quantum Algorithm for the Jones Polynomial}}.
\newblock \eprint{2503.05625}.

\bibitemdeclare{book}{Mitzenmacher_Upfal_2005}
\bibitem{Mitzenmacher_Upfal_2005}
\bibinfo{author}{Michael \surnamestart Mitzenmacher\surnameend} \&
  \bibinfo{author}{Eli \surnamestart Upfal\surnameend} (\bibinfo{year}{2005}):
  \emph{\bibinfo{title}{Probability and Computing: Randomized Algorithms and
  Probabilistic Analysis}}.
\newblock \bibinfo{publisher}{Cambridge University Press}.

\bibitemdeclare{book}{sipser13}
\bibitem{sipser13}
\bibinfo{author}{Michael \surnamestart Sipser\surnameend}
  (\bibinfo{year}{2013}): \emph{\bibinfo{title}{Introduction to the Theory of
  Computation}}, \bibinfo{edition}{third} edition.
\newblock \bibinfo{publisher}{Course Technology}, \bibinfo{address}{Boston,
  MA}.

\bibitemdeclare{inproceedings}{BuhrmanFortnow97}
\bibitem{BuhrmanFortnow97}
\bibinfo{author}{Harry \surnamestart Buhrman\surnameend} \&
  \bibinfo{author}{Lance \surnamestart Fortnow\surnameend}
  (\bibinfo{year}{1997}): \emph{\bibinfo{title}{Resource-bounded kolmogorov
  complexity revisited}}.
\newblock In \bibinfo{editor}{R{\"u}diger \surnamestart Reischuk\surnameend} \&
  \bibinfo{editor}{Michel \surnamestart Morvan\surnameend}, editors: {\slshape
  \bibinfo{booktitle}{STACS 97}}, \bibinfo{publisher}{Springer Berlin
  Heidelberg}, \bibinfo{address}{Berlin, Heidelberg}, pp.
  \bibinfo{pages}{105--116}.

\end{thebibliography}

\appendix

\section{APPENDIX: Proofs} \label{app:proofs}

\cref{thm:noepsilon} states that \cref{def:QKC} can be made independent of $\varepsilon$.
This can be done via amplification of a circuit for  $\varepsilon$ (and time bound $t$).
By the Chernoff bound, this yields a circuit with probability of generating string $x$ that is exponentially close to $1$ at the cost of only a polynomial increase in the time bound.
To apply the bound, we use that any subnormalized list of non-negative numbers has a sufficiently large gap between two probabilities as stated in the following Lemma:
First, we require the following Lemma:

\begin{lemma}\label{thm:gaplem}
Let $p_1 \ge p_2 \ge \cdots \ge p_k \ge p_{k+1} := 0$ be real numbers such that $\sum_{i=1}^k p_i \le 1$. Then there exists an index $1 \le i < k$ for which
$p_i - p_{i+1} \ge \tfrac{p_1^2}{2+p_1}$.
\end{lemma}

\begin{proof}
Without loss of generality (by appending $0$'s or removing probabilities) we may assume that $\tfrac 2 {p_1} \le k \le \tfrac 2{p_1}+1$.
Suppose that for all $1 \le i < k$ we had
$p_i - p_{i+1} < \frac{p_1^2}{2+p_1}$.

Then, since $p_i = p_1 + \sum_{j=1}^{i-1}p_{j+1} - p_{j}$
\begin{align*}
    1 = \sum_i p_i &= kp_1 - \sum_{i=1}^{k-1} i(p_i-p_{i+1})\\
    &> kp_1 - \frac{(k-1)k}{2}\frac{p_1^2}{2+p_1}\\
    &\ge 2 - \frac2{p_1}\frac{2+p_1}{2}\frac{p_1}{2+p_1}
    = 1
\end{align*}
yielding a contradiction.
\end{proof}

{\bf Proof} of \cref{thm:noepsilon}:
\begin{proof}
    Let $U$ be a circuit producing $y$ with probability at least $\varepsilon$ in time bounded by $t$ acting on, say, $N\ge \ell_y$ qubits.
    Let $p_1\ge \dots\ge p_k$ be the non-zero probabilities with which $U$ produces the strings $x_i$, one of which being $y$ (say $x_m = y$).
    By \cref{thm:gaplem} there is $1\le j \le k$ such that the gap between $p_j$ and $p_{j+1}$ is at least $2\delta := \frac{\varepsilon^2}{2 + \varepsilon}$ ($p_{k+1}=0$).
    We can assume that $m\le j$ by applying the lemma only to the probabilities $p_i$ with $m\le i$ (again due to subnormalization).
    We let $L$ be the list of the $j$ most likely strings,  ordered lexicographically, and denote by $a$ the position of $y$ in this list.

    Let $n\in\mathbb{N}$ and take $n$ copies of $U$ denoting by $n_i$ the frequency of $x_i$.
    Each $n_i$ is distributed binomially, so the expected number of occurrences of $x_i$ is $np_i$.
    By the Chernoff bound \cite[Thm.~4.4 \& 4.5]{Mitzenmacher_Upfal_2005} we thus have:
    \begin{align*}
        x_i\not\in L: P[n_i/n \ge (1+\delta_i)p_i]&\le e^{-n\delta_i^2p_i/3}&&,0<\delta_i\\
        x_i\in L: P[n_i/n \le (1-\delta_i)p_i]&\le e^{-n\delta_i^2p_i/2}&&, 0<\delta_i\le1.
    \end{align*}
    Let $\tilde L$ be the list of the $j$ most frequent strings, again ordered lexicographically.
    Letting $\delta_i = \delta/p_i$ and noting that for $x_i\in L$ indeed $\delta_i \le 1$ we find that the probability of $x_i$ being misclassified -- that is $x_i\in \tilde L$ although $x_i\not\in L$ or the other way around -- is bounded from above by $e^{-n\frac{\delta^2}{3p_i}}\le e^{-n(\frac{\delta^2}{3} + \lambda)}$, $\lambda = \frac{\delta^2}3(1-p_1)/p_1$.
    By the union bound we then have
    $$P[\tilde L\neq L] \le ke^{-n (\frac{\delta^2}3+\lambda)}\le e^{-n \frac{\delta^2}3}$$
    for $n\ge \ln(k)/\lambda$.
    Taking the $a$-th element of $\tilde L$ then returns $y$ with probability $\ge 1- r^{-n}$ with $r = e^{\delta^2/3}$.

    Taking $U$ to be a circuit with minimal length description $P_U$, cf.\ \cref{def:QKC}, and concatenating this program with binary descriptions of $n, j, a$ gives a program producing $x$ with probability at least $1-r^{-n}$ in time $nt$.
    
\end{proof}

{\bf Proof} of \cref{rem:icleqC}:

\begin{proof}
Let $P_x$ be the shortest possible program that runs in time $t\leq n^c$, where $c$ is a constant, and generates the $n$-bit string $x$, ie $|P_x|=C^t(x)$.

Consider the following program $Q_x$, which for inputs $y$ of length $|y|=m\leq n$ runs as follows.\\
$Q_x(y)$:\\
-- Run $P_x$ for $m$ steps\\
-- If $P_x$ halts in steps $\leq m$, it generates $P_x=x$\\
----- If $y=x$, then return $Q_x(y) = \chi_\texttt{L}(x)$\\
----- Else if $y\neq x$, then return $Q_x(y)=\bot$\\
-- Else if $P_x$ does not halt in steps $\leq m$, return $Q_x(y)=\bot$

In other words, the program $Q_x$ has the answer to the input instance $x$ hardcoded.
The runtime of $Q_x$ is the runtime of $P_x$ plus the time to compare equality of $y$ and $x$ plus the time to print the single bit $\chi_\texttt{L}(x)$, so we have $t' = t \log t + n + 1$.
The $\log t$ factor comes from simulating the Turing machine that runs $P_x$ using a Turning machine that runs $Q_x$ \cite{sipser13}.
The program $Q_x$ is \texttt{L}-consistent and decides whether $x$ belongs in the language \texttt{L}. Consistency with \texttt{L} is important as without it the bound would be trivial; it would be the length of the program that always just prints a single bit that happens to be $\chi_\texttt{L}(x)$. Its size is $|Q_x|=|P_x| + const$ and so it sets the upper bound for $ic^t(x:L)$.

\end{proof}

{\bf Proof} of \cref{rem:icleqCD}:

\begin{proof}
    This is obvious from the proof of \cref{rem:icleqC} where a trivial program $Q_x$ is considered, which recognises string $x$ as a subroutine.
\end{proof}

To prove \cref{thm:usefulness} we first need the following:
\begin{lemma}\cite{BuhrmanFortnow97}\label{lem:set-chi}
    For any  $A^{\leq n}$ and for all  $x \in A^{\leq n}$:
$CD^{p, A^{\leq n}}(x) \leq 2 \log(\|A^{\leq n}\|) + O(\log n)$ for some polynomial $p$.
\end{lemma}

{\bf Proof} of \cref{thm:usefulness}:

\begin{proof}
Note that since $x \in S^{\leq n}$ and $S$ can be computed by program $P$, we have a $CD^{t+p}$ description of $x$, by Lemma~\ref{lem:set-chi}, of size $m= 2\log(\|S^{\leq n}\|) + O(\log n) +|P|$. Because  $CD^{t'}(x) - ic^t(x:\texttt{L}) = d$ we have that $m \geq |P| + d$. Hence, by Lemma~\ref{lem:set-chi} we have that $S^{\leq n} \geq 2^{(d/2) -c} +O(\log n)$.
\end{proof}

{\bf Proof} of \cref{lemma:pruning}:

\begin{proof}
We prove this by contradiction. Let $\delta < \epsilon$ and assume that for almost all  instances $z$ in \texttt{FAC}: $ic^{t}(z:\texttt{FAC})\leq n^\delta$ for $t(n)\in O(2^{n^\delta})$. We will see that  \factoring can be solved  in time $2^{6n^{\delta}} < 2^{n^\epsilon}$ contradicting the assumption on the hardness of \factoring. 
The idea is to use a pruning algorithm that works as follows. Fix some input $\pair{x,a}$, $n=|x|$. Assume that the input $\pair{x,a}$ is in \fac.  The pruning algorithm will keep track of a set of potential \fac-consistent programs GOOD. Initially GOOD contains all the programs that have size less than or equal $n^\delta$. Observe that the size  $|\text{GOOD}| = 2^{n^\delta}$ which we call $m$. We will also keep track of a set POS of possible  extensions such that there exists a $b\in$ POS such that $\pair{x,ab} \in \fac$.  Initially, we set POS = $\{0,1\}^{n^\delta +1}$, all the possible strings of size $n^\delta +1$. Note that $|\text{POS}| = 2*m$. We run all the programs in GOOD on all the inputs $\pair{x,ab}$ for every $b \in$ POS. Note that if $P$ is a  \fac-consistent program then it must be the case that there is at most one $b\in \text{POS}$ such that $P(\pair{x,ab}) = 1$, because there is exactly one $b$ such that $\pair{x,ab} \in \fac$. So whenever for program $P \in$ GOOD, $P(\pair{x,ab})=1$ and   $P(\pair{x,ab'})=1$ for different $b$ and $b' \in$ POS, we know $P$ is not \fac-consistent and we can remove it from GOOD. 
Since for every $P \in$ GOOD there is at most one $b \in $ POS such that $P(\pair{x,ab}) =1$ we remove all the $b$ from POS for which there is no $P \in$ GOOD such that $P(\pair{x,ab})= 1$. We  have reduced $|\text{POS}|$ by at least half as we have $|\text{POS}| \leq m$. Note that if $\pair{x,a}\in \fac$ then there is a $b \in$ POS such that $\pair{x,ab} \in \fac.$ So this pruning step did not throw away the unique partial extension of $a$ to the largest prime factor of $x$.

Next we append to each of the strings  $b \in$ POS  all $c \in \{0,1\}^k $, for  the minimum $k$ such that $2^k|\text{POS}| \geq 2*m.$ We now repeat the  procedure described above: Run all programs in GOOD on all inputs $\pair{x,ab}$ with $b\in$ POS. Remove the inconsistent ones and reduce $|\text{POS}|$ $\leq m$. Observe that after  each round $|\text{POS}| \leq m$ and that in each round the $|b| \in$ POS grow by at least 1 bit. So after at most $n$ rounds there exist a $b \in$ POS such that $ab$ is a factor of $x$. 

Starting the above  "extend-and-prune" procedure one can find the largest factor $p$ of $x$ by starting it with $\pair{x,\lambda}\footnote{The empty string 
$\lambda$ is a substring, prefix, and suffix of all strings.}$. Lets calculate the time this takes. Each round costs at most $|\text{POS}|*|\text{GOOD}|*t(2n) \leq  4*m*m*t(2n)$, where $t(.)$ is the running time of the individual instance complexity programs. Since we have at most $n$ rounds, the  total running time is upper bounded by  $4*m^2*t(2n)*n \leq  4n2^{4n^\delta}< 2^{n^\epsilon}$.

\end{proof}

\end{document}